\documentclass[10pt,a4paper,oneside,reqno]{amsart}
\usepackage{marginnote}
\usepackage[toc]{appendix}
\usepackage{amssymb}    
\usepackage{mathtools}  
\usepackage{thmtools}
\usepackage{thm-restate}
\usepackage{slashed}
\usepackage{tikz}
\usetikzlibrary{decorations.markings}
\usepackage[
  style=numeric,
  sorting=nyt,
  doi=true,
  url=false,
  isbn=false,
  backend=bibtex
]{biblatex}
\usepackage{xcolor}
\usepackage{hyperref}
\hypersetup{
  colorlinks = true,
  linkcolor  = blue,
  citecolor  = blue,
  urlcolor   = blue}
\usepackage{cleveref}
\newcommand{\reals}{\mathbb{R}}
\newcommand{\complexes}{\mathbb{C}}
\renewcommand{\Re}{\operatorname{Re}}
\renewcommand{\Im}{\operatorname{Im}}

\newcommand{\Dirac}{D}
\declaretheorem[numberwithin=section]{theorem}
\declaretheorem[sibling=theorem]{lemma}
\declaretheorem[sibling=theorem,name=Proposition]{prop}
\declaretheorem[sibling=theorem,style=definition,name=Definition]{definition}
\declaretheorem[sibling=theorem,style=remark]{remark}

\begin{document}

\title{Reflection positivity in Free Fermionic Theories}
\begin{abstract}
    Reflection positivity is one of the Osterwalder--Schrader axioms for Euclidean quantum field theory, ensuring that the reconstructed relativistic Hilbert space carries a positive-definite inner product. For free fermionic theories whose covariance is a real rational function of the Dirac operator, we prove that reflection positivity holds if and only if all poles are real, simple, and carry non-negative residues. We also explicitly show the failure of reflection positivity for covariances with exponential regulators.
\end{abstract}
 \author{Carl Handrack}
 \author{Manfred Salmhofer}
 \address{Institut f\"ur theoretische Physik, Universit\"at Heidelberg, Philosophenweg 19, 69120 Heidelberg}
\maketitle

\section{Introduction and Main Result}

\noindent
Reflection positivity (RP) has been of central importance for quantum field theory ever since the seminal work of Osterwalder and Schrader \cite{OS1973,OS1975}. Its presence assures the existence of a genuine Hilbert space for the quantum mechanical interpretation, together with a natural construction of a Hamiltonian. Moreover, reflection positivity has played a major role in rigorously proving the existence of phase transitions that break continuous symmetries \cite{FSS1976,FILS1,FILS2}. More recently, important work has included the case of Majorana fermions \cite{JaffeMajorana1,JaffeMajorana2} and parafermions~\cite{JaffeParafermions}. 

For free (Gaussian) theories, Wick's theorem reduces the full RP condition to a property of the two-point function alone~\cite{Glimm1987QuantumPhysics,Bogolubov1990GeneralTheory}: the Schwinger functions are reflection positive if and only if the covariance $C$ satisfies
 \begin{equation}\label{eq:RP-def}
     \langle \Theta f,\, C f \rangle_{L^2} \;\geq\; 0 \qquad \forall\, f \in \mathcal{S}(\mathbb{R}^d_+, V),
 \end{equation}
where $\Theta$ denotes time reflection (see Section~\ref{sec:dirac} for precise definitions; we have also included a short proof of this reduction in Section~\ref{sec:reduction}). This makes free theories a natural testing ground for investigating the interplay between the structure of the propagator and the validity of RP.
 
Higher-derivative field theories arise naturally in a number of settings, ranging from questions of regularization to effective field theories and quantum gravity, and the question of their reflection positivity has been the subject of recent work. 
In the scalar case, Arici et al.~\cite{Saueressig2018_RP} showed that for covariances given by real rational functions of the Laplacian, RP holds if and only if all poles are simple, real, and carry non-negative residues --- in other words, precisely when the propagator decomposes into a positive combination of standard Klein--Gordon propagators. The main purpose of the present note is to extend this result to free fermions, where the Laplacian is replaced by the Dirac operator $\slashed{\partial}$ (defined below in Section~\ref{sec:dirac}). By adapting the contour integration strategy of~\cite{Saueressig2018_RP} to the algebraic properties of the Dirac operator, we prove the following theorem. 

\begin{theorem}\label{thm:main}
Let $C$ be a real rational function without poles on $i\mathbb{R}$. Then the covariance $C(\slashed{\partial})$ is reflection positive if and only if all poles of $C$ lie in $\mathbb{R}$ and are simple and have non-negative residue.
\end{theorem}

\begin{remark}
This implies that each of the poles corresponds to
a free propagator with positive squared mass, like in the scalar case. 
\end{remark}

The paper is organized as follows. Section~\ref{sec:dirac} introduces the Dirac covariance and the class of rational functions under consideration. Section~\ref{sec:Properties of IC} collects the properties of the associated quadratic form needed for the proof. In Section~\ref{sec:proof}, the main theorem is proved, treating sufficiency and necessity in turn, while in Section~\ref{sec:reduction} the equivalence between reflection positivity of the covariance and Osterwalder--Schrader positivity of the full Gaussian theory is established for complex scalar fields and Dirac fermions. The analogous equivalence for real scalar fields is proved in the first appendix, while the second appendix discusses exponential regulators.

\section{Setup, Notation and Reflection Positivity}\label{sec:dirac}

\subsection{The Dirac covariance}
Let $(V, \langle\cdot,\cdot\rangle_V)$ be a complex, finite-dimensional Hilbert space carrying a self-adjoint representation of the Clifford algebra generated by $\{\gamma^i, \gamma^j\}=2\delta^{ij}$, and write $\slashed{p} = \gamma^i p_i$, $\slashed{\partial} = \gamma^i \partial_i$. With the convention $\Delta=-\partial_i\partial^i\geq 0$, one has
\begin{equation}
\slashed\partial^2=-\Delta,
\qquad
-(\slashed\partial+m)(\slashed\partial-m)=\Delta+m^2 .
\end{equation}
For $m>0$, we define the Dirac covariance to be the translation-invariant pseudodifferential operator (cf.~\cite[Ch.~XVIII]{Hormander2007AnalysisIII}) on \(\mathcal S(\mathbb R^d,V)\) with symbol
\begin{equation}
\widehat \Dirac(p)=\frac{i\slashed p-m}{p^2+m^2},
\end{equation}
that is, $\Dirac$ acts on spinor-valued test functions $f\in \mathcal{S}(\mathbb R^d,V)$ by
\begin{equation}
    (\Dirac f)(x) = \int_{\reals^{d}}\frac{\mathrm{d}^{d} p}{(2\pi)^{d/2}} \;e^{ip\cdot x}\frac{i\slashed p-m}{p^2+m^2}\, \hat{f}(p),
\end{equation}
The convention for the Fourier transform used here is
\begin{equation}
    \hat{f}(p) = \int_{\reals^{d}}\frac{\mathrm{d}^{d} x}{(2\pi)^{d/2}} \;e^{-ip\cdot x} f(x), \quad f(x) = \int_{\reals^{d}}\frac{\mathrm{d}^{d} p}{(2\pi)^{d/2}} \;e^{ip\cdot x} \hat{f}(p).
\end{equation}

Lastly, under Euclidean time reflection
\begin{equation}
    \vartheta: \mathbb{R}^d \to \mathbb{R}^d, \qquad (x^0, \vec{x}) \mapsto (-x^0, \vec{x}),
\end{equation}
spinor-valued functions transform as 
\begin{equation}
    (\Theta f)(x) = \gamma^0 f(\vartheta x). 
\end{equation} $\Theta$ is an involution on $\mathcal{S}(\reals^d,V) $~\cite{Jaffe2008ReflectionMonotonicity}.

\subsection{Reflection positivity}
In the following, we write $\langle \cdot, \cdot \rangle_{L^2}$ for the standard inner product on the Hilbert space $L^2(\mathbb{R}^d, V)$, and $\mathcal{S}(\mathbb{R}^d_+, V)$ for the subspace of test functions supported in the positive-time half-space $\mathbb{R}^d_+ = \{x \in \mathbb{R}^d \,\vert\, x^0 > 0\}$.
\begin{definition}[Reflection positivity]\label{def:RP}
For a covariance operator $C$ on $\mathcal{S}(\reals^d, V)$ and test functions $f, g \in \mathcal{S}(\mathbb{R}^d_+, V)$, define the sesquilinear form
\begin{equation}\label{eq:OS-form}
    Q_C[f, g] = \langle \Theta f,\, C\, g \rangle_{L^2},
\end{equation}
and $I_C[f] = Q_C[f, f]$. We call $C$ \emph{reflection positive} if $Q_C$ is hermitian and non-negative, i.e.
\begin{equation}\label{eq:RP}
    I_C[f] \;\geq\; 0 \qquad \forall\, f \in \mathcal{S}(\mathbb{R}^d_+, V).
\end{equation}
\end{definition}

\begin{remark}\label{rem:hermiticity}
The form $Q_C$ is hermitian if and only if $\Theta\, C\, \Theta^{-1} = C^*$, where $C^*$ denotes the formal $L^2$-adjoint of $C$ on $\mathcal{S}(\mathbb{R}^d, V)$. In Fourier space, this reads $\smash{\widehat{C}(-p_0, \vec{p}) = \gamma^0\, \widehat{C}(p)^\dagger\, \gamma^0}$, where $\dagger$ is the matrix adjoint on $V$. Hermiticity ensures that $I_C[f] \in \reals$ for all $f$, which is necessary for the positivity condition~\eqref{eq:RP} to be meaningful.
\end{remark}

\subsection{Higher-derivative covariances}
From now on, we consider covariances of the form
$C(\slashed{\partial})$, where $C$ is a real rational function
with no poles on $i\mathbb{R}$. Both conditions will be standing
assumptions throughout. The absence of poles on
$\operatorname{spec}(\slashed{\partial})=i\mathbb{R}$ ensures that
$C(\slashed{\partial})$ is well-defined, while reality, meaning that
\begin{equation}
    C(\overline{z})=\overline{C(z)},
\end{equation}
guarantees that $Q_{C(\slashed{\partial})}$ is hermitian. In
particular, the nonreal poles of $C$ occur in conjugate pairs.

Let $\mathcal{P}_C$ denote the set of $N$ distinct poles of $C$. The partial fraction decomposition of $C$ reads
\begin{equation}\label{eq:partial-frac}
    C(x) = \sum_{j=1}^{N} \sum_{n=1}^{k_j} a_{jn}(x - z_j)^{-n} + p(x), \qquad z_j \in \mathcal{P}_C, \quad a_{jn} \in \mathbb{C}, \quad k_j \in \mathbb{N},
\end{equation}
with a real polynomial $p$.
Since $\mathcal{P}_C \cap i\mathbb{R}=\emptyset$, the Fourier multipliers $(i\slashed p-z_j)^{-n}$ are smooth symbols with polynomially bounded derivatives, hence act continuously on $\mathcal S(\mathbb R^d,V)$. Consequently, $C(\slashed\partial)$ defines a continuous linear operator on $\mathcal{S}(\mathbb{R}^d,V)$, given by
\begin{equation}\label{eq:partial-frac-op}
    C(\slashed{\partial}) = \sum_{j=1}^{N} \sum_{n=1}^{k_j} a_{jn}(\slashed{\partial} - z_j)^{-n} + p(\slashed{\partial}).
\end{equation}
For $f \in \mathcal{S}(\mathbb{R}^d_+, V)$, the polynomial term does not contribute to $I_C[f]$. Since $p(\slashed{\partial})$ is local, $\operatorname{supp} p(\slashed{\partial})f \subseteq \mathbb{R}^d_+$, while $\operatorname{supp} \Theta f \subseteq \mathbb{R}^d_-$, and the two supports are disjoint. Thus,
\begin{equation}\label{eq:IC-decomp}
    I_C[f] = \sum_{j=1}^{N} \sum_{n=1}^{k_j} a_{jn}\, I_{z_j}^{(n)}[f], \qquad I_z^{(n)}[f] \;=\; \langle \Theta f,\, (\slashed{\partial} - z)^{-n} f \rangle_{L^2},
\end{equation}
for any $z \in \mathbb{C} \setminus i\mathbb{R}$. We abbreviate $I_z[f] = I_z^{(1)}[f]$.
Note that $I_{\bar{z}}^{(n)}[f] = \overline{I_z^{(n)}[f]}$, so conjugate poles contribute complex conjugate terms to~\eqref{eq:IC-decomp}. This motivates writing
\begin{equation}\label{eq:IC-real-complex}
I_C[f] = \sum_{j:z_j \in\reals} \sum_{n=1}^{k_j} a_{jn}\, I_{z_j}^{(n)}[f] \;+\; 2\Re \sum_{j:\Im z_j > 0} \sum_{n=1}^{k_j} a_{jn}\, I_{z_j}^{(n)}[f],
\end{equation}
which makes it manifest that $I_C[f] \in \mathbb{R}$. We write $C_\lambda(\slashed{\partial})$ for the combined contribution to~\eqref{eq:partial-frac-op} from a pole $\lambda$ with $\operatorname{Im}\lambda \geq 0$ and, when $\lambda \notin \mathbb{R}$, its conjugate $\bar{\lambda}$, so that~\eqref{eq:IC-real-complex} can be written compactly as
\begin{equation}
    I_C[f] = \sum_{\substack{\lambda \in \mathcal{P}_C \\ \operatorname{Im}\lambda \geq 0}} \langle \Theta f,\, C_\lambda(\slashed{\partial})\, f \rangle.
\end{equation}

\section{Properties of the Quadratic Form \texorpdfstring{$I_C$}{I}}\label{sec:Properties of IC}

\noindent
We proceed by collecting several properties of the quadratic form $I_C$.
\subsection{Analyticity in half-planes and homogeneity}

\begin{lemma}\label{derivative formulation}
For $z_j \in \mathcal{P}_C \subset \mathbb{C} \setminus i\mathbb{R}$, the form $I_C$ can be expressed entirely in terms of derivatives of $I_z$:
\begin{equation}
    I_C[f] =  \sum_{j=1}^{N} \left(\sum_{n=1}^{k_j}\frac{a_{jn}}{(n-1)!}  \frac{\mathrm{d}^{n-1}}{\mathrm{d}z^{n-1} } {I_{z}[f]}\Bigr|_{z=z_j}\right). 
\end{equation}
\end{lemma}
\begin{proof} 
By~\eqref{eq:IC-decomp} 
\begin{equation}
    I_C[f] = \sum_{j=1}^{N} \sum_{n=1}^{k_j} a_{jn}\, I_{z_j}^{(n)}[f],
\end{equation} with $I_z^{(n)}[f] \;=\; \langle \Theta f,\, (\slashed{\partial} - z)^{-n} f \rangle$. The key identity is
\begin{equation}
    \frac{1}{(n-1)!}\frac{\mathrm{d}^{n-1}}{\mathrm{d} z^{n-1}} (-i\slashed{p} - z)^{-1} =  (-i\slashed{p} - z)^{-n}, \qquad z \in \mathbb{C} \setminus i\mathbb{R}.
\end{equation}
Applying this to the Fourier representation of $I_z[f]$ (see also \Cref{reduction to on shell}), and using absolute convergence of the integral, we may exchange differentiation and integration. 
\end{proof}
The following observation is key to all the analysis to follow, as it allows us to link reflection positivity to the analytic structure of the covariance. Since $f \in \mathcal{S}(\mathbb{R}^d_+, V)$ is supported only where $x^0>0$, the Paley--Wiener theorem, see \cite[Theorem 7.2.4]{Strichartz1994ATransforms} or \cite[Sec. 9]{Rudin1987RealAnalysis.}, guarantees that the Fourier transform $p_0 \mapsto \hat{f}(p_0, \vec{p})$ extends to a holomorphic function in the lower half plane for each fixed $\vec{p}$.
\begin{lemma}\label{reduction to on shell}
 For a mass $z \in \complexes\setminus i\reals$,  $I_z[f]$ has the following representation in Fourier space:
    \begin{equation}
        I_z[f] = \pi\int_{\mathbb{R}^{d-1}} \!\! \frac{\mathrm{d}^{d-1}\vec{p}}{\omega_{\vec{p}}} \left\langle \hat{f}(-i{\bar\omega}_{\vec{p}},\vec{p}), \gamma^0  (i\slashed p - z)\hat{f}(-i\omega_{\vec{p}},\vec{p})\right\rangle_V, 
    \end{equation} 
   where $\omega_{\vec{p}} = \sqrt{z^2 + |\vec{p}|^2}$ (principal branch, so that $\Re\omega_{\vec{p}} > 0$ and $-i\omega_{\vec{p}}$ lies in the lower half plane $\mathbb{H}_-$), $\slashed p$ is evaluated at $p_0 = -i\omega_{\vec{p}}$ and $\hat{f}$ denotes the Fourier transform of $f$.
\end{lemma}
\begin{proof}
Passing to Fourier space and using that $\widehat{(\Theta f)}(p) = \Theta\hat{f}(p)$,
\begin{align}
    I_z[f] 
    &= \int_{\reals^{d}}\!\! \mathrm{d}^{d} x\; \Big\langle (\Theta f)(x), ((\slashed{\partial} - z)^{-1}f)(x)\Big\rangle_V\\
    &= \int_{\mathbb{R}^d}\!\! \mathrm{d}^d p\; \Big\langle \hat{f}(\vartheta p),\; \gamma^0 (-i\slashed{p} - z)^{-1}\, \hat{f}(p) \Big\rangle_V \\
    &= \int_{\mathbb{R}^d}\!\! \mathrm{d}^d p\; \Big\langle \hat{f}(\vartheta p),\; \gamma^0 \frac{i\slashed{p} - z}{p^2 + z^2}\, \hat{f}(p) \Big\rangle_V. \label{eq:Iz-fourier}
\end{align}
Because the denominator is nonvanishing and the Schwartz functions decay rapidly, the integral is absolutely convergent, so by Fubini's theorem
\begin{equation}
    I_z[f] = \int_{\reals^{d-1}}\mathrm{d}^{d-1} p \int_{\reals}\mathrm{d} p_0\;\Big\langle   \hat{f}(-p_0,\vec p),\gamma^0\frac{i\slashed p - z}{p_0^2 +|\vec p|^2 +z^2}\hat{f}(p_0, \vec p)\Big\rangle_V
\end{equation}
Note that $\hat{f}(\zeta, \vec{p})$ extends holomorphically onto $\mathbb{H}_-$ by the Paley--Wiener theorem. Because  $\langle.,.\rangle_V$ is antilinear in its first argument, we use that $\hat{f}(-p_0,\vec{p})$ extends to the anti-holomorphic function $\hat{f}(-\bar{\zeta}, \vec{p})$ on the same half-plane, so that the integrand as a whole is meromorphic on $\mathbb{H}_-$.

Closing the $p_0$-contour to $\Gamma_R$ with a large semicircle in the lower half plane (\Cref{Tikz integration contour}), whose contribution vanishes as $R \to \infty$, the residue theorem gives
\begin{align}
    I_z[f] &= \int_{\reals^{d-1}}\mathrm{d}^{d-1} \vec{p} \lim_{R\to\infty}\oint_{\Gamma_R}\mathrm{d} \zeta \Big\langle  \hat{f}(-\bar{\zeta}, \vec{p}),\gamma^0\frac{i\slashed p - z}{\zeta^2+(|\vec{p}|^2+z^2)}\hat{f}(\zeta, \vec{p})\Big\rangle_V\\
    &=\int_{\reals^{d-1}}\mathrm{d}^{d-1} \vec{p} \; (\text{-}2\pi i) \underset{\zeta\to \text{-}i\omega_{\vec{p}}}{\mathrm{Res}} \Big\langle \hat{f}(-\bar{\zeta}, \vec{p}),\gamma^0 \frac{i\slashed p - z}{\zeta^2+\omega_{\vec{p}}^2}\hat{f}(\zeta, \vec{p})\Big\rangle_V\\
    &= \pi \int_{\reals^{d-1}}\frac{\mathrm{d}^{d-1} \vec{p}}{\omega_{\vec{p}}}\;
    \left\langle  \hat{f}(-i\bar{\omega}_{\vec{p}},\vec{p}),\gamma^0 (i\slashed p - z)\hat{f}\left(-i\omega_{\vec{p}}, \vec{p}\right)\right\rangle_V,
 \end{align}
where we have picked up the residue at $-i\omega_{\vec{p}}$ in the lower half plane,
\begin{equation}
\underset{\zeta\to \text{-}i\omega_{\vec{p}}}{\mathrm{Res}} \left(\frac{1}{\zeta^2+\omega_{\vec{p}}^2}\right)= -\frac{1}{2i \omega_{\vec{p}}} \;,
\end{equation}
and where $\gamma^0 (i\slashed p - z) = i \omega_{\vec{p}} + i \sum_{\mu > 0} \gamma_\mu p_\mu - z \gamma_0$. 
An additional minus sign in the second step originates from the negative orientation of our integration contour (see \Cref{Tikz integration contour}).
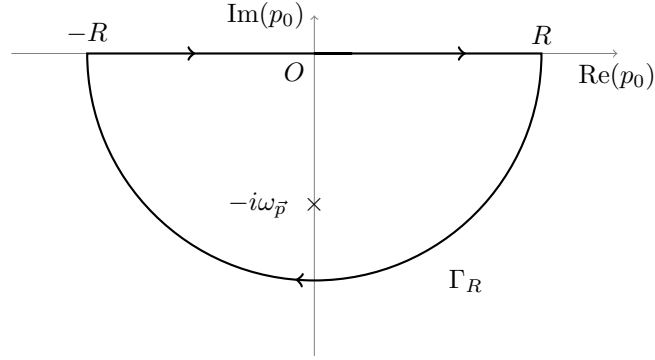
\begin{figure}
\centering
\begin{tikzpicture}
  [
    decoration={
      markings,
      mark=at position 2cm with {\arrow[line width=1pt]{>}},
      mark=at position 0.5 with {\arrow[line width=1pt]{>}},
      mark=at position 0.87 with {\arrow[line width=1pt]{>}},
    }
  ]
  \draw [help lines,->] (-4,0) -- (4,0) coordinate (xaxis);
  \draw [help lines,->] (0,-4) -- (0,0.5) coordinate (yaxis);
  \node at (0,-2) {$\times$};
  \node at (-0.75,-2) {$-i\omega_{\vec{p}}$};
  \path [draw, line width=0.8pt, postaction=decorate] (0,0)  -- (3,0) node [above] {$R$} arc (0:-180:3) node [above] {$-R$} -- (.5,0);
  \node [below] at (xaxis) {$\mathrm{Re}(p_0)$};
  \node [left] at (yaxis) {$\mathrm{Im}(p_0)$};
  \node [below left] {$O$};
  \node at (2,-3) {$\Gamma_{R}$};
\end{tikzpicture}
\caption{The integration contour $\Gamma_R$ in the lower half plane with a negative orientation. In general, $-i\omega_{\vec{p}}$ need not lie directly on the imaginary axis.}
\label{Tikz integration contour}
\end{figure}
\end{proof}
With this, we can prove the following homogeneity property of $I_z[f]$, which will be used to construct explicit example functions that violate reflection positivity:
\begin{lemma}[Homogeneity Property]\label{Homogeneity Lemma} Let $f \in \mathcal{S}(\reals^d_+, V)$ and $q$ be a polynomial. Then, $q(\Delta)f \in \mathcal{S}(\mathbb{R}^d_+, V)$ and the following homogeneity property applies:
\begin{equation}\label{Homogeneity Lemma Equation}
    I_z [q(\Delta)f] = \overline{q(-\bar{z}^2)}q(-z^2) I_z
    [f]
\end{equation}
\end{lemma}
\begin{proof} Since $q(\Delta)$ is polynomial in the local operator $\Delta$, it is itself a local operator and thus $\mathrm{supp}(q(\Delta) f) \subseteq \mathrm{supp} f$. As $q(\Delta) f$ is also just a sum of derivatives of $f$, it is also a Schwartz function. Together, this implies $q(\Delta) f \in \mathcal{S}(\reals^d_+, V)$.

The second claim follows directly from \Cref{reduction to on shell} applied to $q(\Delta)f$, and using that $\widehat{q(\Delta)f}(p)=q(p^2)\hat f(p)$ and $(-i\omega_{\vec p})^2 +\lvert\vec p\rvert^2 = -z^2$.

\end{proof}
\subsection{Reduction to separate poles}
The proof of necessity in \Cref{thm:main} requires constructing test functions $f$ for which $I_C[f] < 0$. The following lemma reduces this construction to one pole at a time: for any prescribed pole, f can be chosen so that the contributions from all other poles vanish.

\begin{lemma}[Reduction to Separate Poles]\label{reduction lemma}
    For any pole $\lambda \in \mathcal P_C$, there exists $f \in \mathcal{S}(\reals^d_+, V)$ such that 
    \begin{equation}
        I_C[f] 
        = \langle \Theta f,C_\lambda(\slashed \partial) f\rangle_{L^2}
    \end{equation}
    and $I_C[f] \neq 0$.
\end{lemma}
\begin{proof} Without loss of generality, we may assume that $\Im \lambda\geq 0$. Let $f_0$ be any function in  $\mathcal{S}(\reals^d_+, V)$ such that the integral $I_\lambda[f_0]$ does not vanish. Choose the polynomial
\begin{equation}
	q(X)  = \prod_{z_j\in \mathcal P_C\setminus \{\lambda, \Bar{\lambda}\}} (-X-z_j^2)^{k_j}.
\end{equation}
By \Cref{derivative formulation},
\begin{equation}
	I_C[q(\Delta)f] =  \sum_{j=1}^{N} \left(\sum_{n=1}^{k_j}\frac{a_{jn}}{(n-1)!} 	\frac{\mathrm{d}^{n-1}}{\mathrm{d}z^{n-1} } I_{z}[q(\Delta)f]\Bigr|_{z=z_j}\right). 
\end{equation}
For each contribution from a pole $z_j \notin \{\lambda, \Bar{\lambda}\}$, define the polynomial $q_j(X) = q(X)/(-X-z_j^2)^{k_j}$ and use the homogeneity property to write
\begin{align}
&\sum_{n=1}^{k_j}\frac{a_{jn}}{(n-1)!} 	\frac{\mathrm{d}^{n-1}}{\mathrm{d}z^{n-1} } I_{z}[q(\Delta)f]\Bigr|_{z=z_j}\\
&= \sum_{n=1}^{k_j}\frac{a_{jn}}{(n-1)!} 	\frac{\mathrm{d}^{n-1}}{\mathrm{d}z^{n-1}}(z^2-z_j^2)^{k_j}\overline{(\bar{z}^2-z_j^2)^{k_j}} I_{z}[q_j(\Delta)f]\Bigr|_{z=z_j} \\
&= \sum_{n=1}^{k_j}\frac{a_{jn}}{(n-1)!} 	\frac{\mathrm{d}^{n-1}}{\mathrm{d}z^{n-1}}(z^2-z_j^2)^{k_j}{(z^2-\bar{z}_j^2)}^{k_j} I_{z}[q_j(\Delta)f]\Bigr|_{z=z_j}. 
\end{align}
Because $(z^2-z_j^2)^{k_j} =(z-z_j)^{k_j}(z+z_j)^{k_j} $ and $n-1 <k_j$, at least one factor $(z-z_j)$ remains in each term, such that the contribution vanishes upon evaluation at $z= z_j$.   
Thus, the surviving terms are precisely
\begin{equation}
    I_C[q(\Delta)f] = \langle \Theta f,\, C_\lambda(\slashed{\partial})\, f \rangle_{L^2} \neq 0,
\end{equation}
where the non-vanishing follows from the homogeneity property, which allows one to extract $q(-\bar{z}^2)q(-z^2)$ as a non-zero factor.
\end{proof}

\section{Proof of Main Theorem}\label{sec:proof}
With the properties of $I_C$ established in \Cref{sec:Properties of IC}, we are now in a position to prove \Cref{thm:main}. We start with the sufficiency of the conditions stated in the theorem, which is the easier direction.
\subsection{Sufficient conditions for reflection positivity}

\begin{prop}\label{sufficiency}
The operator $C(\slashed{\partial})$ for a rational function $C$ is reflection positive if $\mathcal{P}_C \subseteq \mathbb{R} \setminus \{0\}$ and the poles are simple with non-negative residues.
\end{prop}
\begin{proof}
By \Cref{reduction to on shell}, the form $I_\mu[f]$ for a single propagator $(\slashed{\partial} - \mu)^{-1}$ is given by
\begin{equation}
    I_\mu[f] = \pi \int_{\mathbb{R}^{d-1}} \frac{\mathrm{d}^{d-1}p}{\omega_{\vec{p}}} \left\langle \hat{f}(-i\bar{\omega}_{\vec{p}}, \vec{p}),\, \gamma^0(i\slashed{p} - \mu)\, \hat{f}(-i\omega_{\vec{p}}, \vec{p}) \right\rangle_V.
\end{equation}
For $\mu \in \mathbb{R}$, we have $\bar{\omega}_{\vec{p}} = \omega_{\vec{p}}$, so the integrand reduces to $\langle v_p,\, \gamma^0(i\slashed{p} - \mu)\, v_p \rangle_V$ with $v_p = \hat{f}(-i\omega_{\vec{p}}, \vec{p}) \in V$. It therefore suffices to show that $\gamma^0(i\slashed{p} - \mu)$ is positive semi-definite at $p_0 = -i\omega_{\vec{p}}$. Writing $\gamma^0(i\slashed{p} - \mu) = \omega_{\vec{p}} I + \Omega$ with $\Omega = i\gamma^0 \vec{\gamma} \cdot \vec{p} - \mu\gamma^0$, one verifies that $\Omega$ is Hermitian with $\Omega^2 = \omega_{\vec{p}}^2 I$, so $\Omega$ has eigenvalues $\pm\omega_{\vec{p}}$ and hence $\omega_{\vec{p}} I + \Omega$ has eigenvalues $0$ and $2\omega_{\vec{p}} \geq 0$.

For a general $C(\slashed{\partial})$ satisfying the stated conditions, all poles are real and simple, so~\eqref{eq:IC-decomp} gives
\begin{equation}
    I_C[f] = \sum_{j=1}^{N} a_j\, I_{\mu_j}[f] \geq 0. \qedhere
\end{equation}
\end{proof}

\subsection{Necessary conditions for reflection positivity}
\begin{prop}[Complex Poles]\label{Complex Poles Prop}
	Assuming $C$ has a genuinely complex pole $\lambda \in \complexes \setminus i\reals$, there exists a function $f\in \mathcal{S}(\mathbb{R}^d_+, V)$ such that 
    \begin{equation*}
    I_C[f] < 0.
    \end{equation*}
\end{prop}
\begin{proof}
By \Cref{reduction lemma} there  exists $f'\in \mathcal{S}(\mathbb{R}^d_+, V)$ such that $C$ only has two complex poles contributing to $I_C[f']\neq 0$. These poles of order $k$ lie conjugate to each other, and $\lambda \neq \bar \lambda$. We choose $\lambda$ to have $\mathrm{Im}(\lambda^2) > 0$. This is always possible, since by our assumption $\mathcal P_C \cap i\reals = \emptyset$.\\
Motivated by \cite{Saueressig2018_RP}, we make the polynomial ansatz $f=q(\Delta) f'$, where
\begin{equation}
    q(z) = (-z-\lambda^2)^{k-1} h(-z),\;\;\;\; h(z) = \frac{1-\bar{\alpha}}{2i\mathrm{Im}(\lambda^2)}(z-\lambda^2)+1, \; \alpha \in \complexes.
\end{equation}
In contrast to the polynomial chosen for the scalar case, we have chosen the polynomial $h$ to satisfy $h(\lambda^2) =1$ and $\overline{h(\bar{\lambda}^2)}= \alpha$. This modification compared to \cite[Proposition III.5.]{Saueressig2018_RP} is necessary because the homogeneity property only allows us to evaluate at the (negative) square of the mass $\lambda$. \Cref{derivative formulation} and \Cref{Homogeneity Lemma} imply
\begin{align*}
I_C[q(\Delta)f'] &= \sum_{n=1}^{k}\frac{a_{n}}{(n-1)!} 	\frac{\mathrm{d}^{n-1}}{\mathrm{d} z^{n-1} } I_{z}[q(\Delta)f']\Bigr|_{z=\lambda} + c.c.\\
&= 2\; \mathrm{Re} \sum_{n=1}^{k}\frac{a_{n}}{(n-1)!} 	\frac{\mathrm{d}^{n-1}}{\mathrm{d} z^{n-1} }q(-z^2)\overline{q(-\bar{z}^2)} I_{z}[f']\Bigr|_{z=\lambda}\\
&= 2\; \mathrm{Re} \sum_{n=1}^{k}\frac{a_{n}}{(n-1)!} 	\frac{\mathrm{d}^{n-1}}{\mathrm{d} z^{n-1} }(z^2-\lambda^2)^{k-1}h(z^2)\overline{q(-\bar{z}^2)} I_{z}[f']\Bigr|_{z=\lambda}\\
&= 2\; \mathrm{Re} \sum_{n=1}^{k}\frac{a_{n}}{(n-1)!} 	\frac{\mathrm{d}^{n-1}}{\mathrm{d} z^{n-1} }(z-\lambda)^{k-1}(z+\lambda)^{k-1}h(z^2)\overline{q(-\bar{z}^2)} I_{z}[f']\Bigr|_{z=\lambda}.
\end{align*}The only term remaining carries the index $n=k$, for which precisely all the factors of $(z-\lambda)$ are removed by derivation (all other terms generated when applying the product rule will vanish upon evaluation due to the remaining factors of $(z-\lambda)$).\\Using $\overline{q(-\bar{\lambda}^2)}= \alpha(2 i \mathrm{Im}(\lambda^2))^{k-1}$ as well as $h(\lambda^2)= 1$, we arrive at
\begin{align*}
I_C[q(\Delta)f'] &= 2\;\mathrm{Re} \left(a_n (\lambda+\lambda)^{k-1} h(\lambda^2)\overline{q(-\overline{\lambda}^2)}I_\lambda[f'] \right)\\
&= 2^{k}\mathrm{Re} \left(a_n \cdot\alpha \cdot\lambda^{k-1}(2i\mathrm{Im}(\lambda^2))^{k-1}I_\lambda[f'] \right).
\end{align*}
This expression does not vanish identically and changes sign when $\alpha$ does, so that we can find an $\alpha$ such that this equation becomes negative.
\end{proof}

\begin{prop}[Higher Order Poles] 
Assume now that $C$ has a \textit{real} pole $\mu \ne 0$ of order $k>1$. Then there exists $f \in \mathcal{S}(\mathbb{R}^d_+, V)$ such that 
\begin{equation*}
I_C[f] < 0.
\end{equation*}
\end{prop}
\begin{proof}This construction follows the same strategy as that used by \textcite{Saueressig2018_RP} for the scalar case save for a slight modification of the polynomial ansatz.

Let $f_0 \in \mathcal{S}(\mathbb{R}^d_+, V)$ satisfy $I_\mu[f_0] \neq 0$, and define
\begin{equation}
        f_1 =
    \prod_{z_j \in \mathcal{P}_C\setminus\{\mu\}}
    (-\Delta - z_j^2)^{k_j} f_0
\end{equation}
to simultaneously annihilate all poles except $\mu$ and preserve the non-vanishing of $I_\mu[f_1]$.
Now, take the real polynomial
\begin{equation}
    q(X) = \alpha (-X-\mu^2)^{k-1} +1, \qquad \alpha \in \reals,
\end{equation}
so that 
\begin{align}
    &I_C[q(\Delta)f_1] = \sum_{n=1}^{k}\frac{a_n}{(n-1)!} 	\frac{\mathrm{d}^{n-1}}{\mathrm{d} z^{n-1} } I_{z}[q(\Delta)f_1]\Bigr|_{z=\mu}\\
    &=  \sum_{n=1}^{k}\frac{a_n}{(n-1)!} \frac{\mathrm{d}^{n-1}}{\mathrm{d}z^{n-1} } \left(1+ 2\alpha(z^2-\mu^2)^{k-1}+\alpha^2(z^2-\mu^2)^{2(k-1)}\right)I_z[f_1] \Big\rvert_{z=\mu}\\
    &= I_C[f_1] + 2\alpha\cdot a_k \cdot(2\mu)^{k-1}I_{\mu}[f_1].
\end{align}
Since $I_\mu[f_1] \neq 0$, the right-hand side is a non-constant affine function in $\alpha$, which can be made negative by choosing $\alpha$ appropriately.
\end{proof}
\begin{prop}[Negative Residues]\label{Negative Residues}
    If $C$ has one pole with negative residue, then $C(\slashed{\partial})$ is not reflection positive.
\end{prop}
\begin{proof}
For the remaining case of one real pole with negative residue, one can use that $(\slashed{\partial} - \mu)^{-1}$ is reflection positive, so that the additional sign from the negative residue implies $I_\mu[f] \leq 0$. 
\end{proof}
\noindent
From all this, \Cref{thm:main} now follows immediately.
\bigskip

\section{Reduction to the Two-Point Function}\label{sec:reduction}
For interacting theories, Osterwalder–Schrader positivity generally constrains the full hierarchy of Schwinger functions and cannot be inferred from the two-point function alone~\cite{OS1975}. For free theories, however, Wick’s theorem makes positivity of the full hierarchy equivalent to reflection positivity of the covariance. We prove this reduction for bosons ($+$) and fermions ($-$) in parallel. Internal indices will be suppressed; in the fermionic case,
all test functions are implicitly understood to be spinor-valued.

Denote by
\begin{equation}
    \mathfrak B = \bigoplus_{n\ge 0} \mathcal S(\reals^d)^{\hat\otimes n}, \qquad  \mathcal S(\reals^d)^{\hat\otimes 0} = \complexes
\end{equation}
the Borchers algebra \cite{Borchers1962} of test functions. Its elements are sequences
\begin{equation}
    F=(f_n)_{n\ge 0}, \qquad f_n\in \mathcal S(\reals^d)^{\hat\otimes n} 
\end{equation}
with only finitely many nonzero components. Here $\hat\otimes$ denotes the completed
topological tensor product and $\bigoplus$ the algebraic direct sum. The subspace of positive-time test functions is defined analogously as
\begin{equation}
    \mathfrak B_+ = \bigoplus_{n\ge 0} \mathcal S(\reals^d_+)^{\hat\otimes n}.
\end{equation}

The time reflection $\Theta$ extends from the one-particle test functions to tensor powers. On decomposable tensors, our convention is
\begin{equation}
    \Theta(f_1\otimes\dots\otimes f_n) =
    \begin{cases}
        \Theta f_1\otimes\cdots\otimes\Theta f_n,
        & \text{for bosons}, \\[2mm]
        \Theta f_n\otimes\cdots\otimes\Theta f_1,
        & \text{for fermions}.
  \end{cases}
\end{equation}

For $m,n\in\mathbb N_0$, let
\begin{equation}
  S_{m,n}^{\pm}
  \in
  \mathcal S'\bigl((\mathbb R^d)^m\times(\mathbb R^d)^n\bigr)
\end{equation}
be the corresponding Schwinger distributions.

\begin{definition}[Osterwalder-Schrader Positivity]\label{def:OS-positivity}
    The family $(S^\pm_{n,m})_{n,m\in\mathbb N_0}$ is
    \emph{Osterwalder--Schrader positive} if the sesquilinear form
\begin{equation}
    (F,G)_{\mathrm{OS}}
    =
    \sum_{m,n\ge 0} S^\pm_{n,m}(\Theta \bar{f}_m\otimes g_n),
    \quad
    F=(f_n),\, G=(g_n)\in \mathfrak B_+,    
\end{equation}
is positive semidefinite. 
 
\end{definition}
For complex scalar fields and Dirac fermions, $U(1)$-invariance implies
\begin{equation}
    S^\pm_{m,n}=0 \qquad \text{if } n\neq m,
\end{equation}
and we abbreviate $S^\pm_{n,n}=S^\pm_{2n}$.

By Wick's theorem
\cite[Sections 8.4A and 8.4C]{Bogolubov1990GeneralTheory}
\cite[Section 8.3]{Glimm1987QuantumPhysics}, we have 
\begin{equation}
  S_{2n}^{\pm}
  \bigl(
    \bar{f}_1\otimes\cdots\otimes\bar{f}_n
    \otimes g_n\otimes\cdots\otimes g_1
  \bigr)
  =
  \begin{cases}
    \operatorname{perm}
    \bigl(\langle f_i,Cg_j\rangle\bigr)_{i,j=1}^n,
      & \text{for bosons}, \\[2mm]
    \det
    \bigl(\langle f_i,Cg_j\rangle\bigr)_{i,j=1}^n,
      & \text{for fermions}
  \end{cases}
\end{equation}
 for $f_1,\ldots,f_n,g_1,\ldots,g_n\in\mathcal S(\mathbb R^d)$. We write this more compactly as $\det^-=\det$ and $\det^+ = \operatorname{perm}$. Consequently, Wick’s theorem reduces Osterwalder–Schrader positivity of the free theory to reflection positivity of the covariance $C$.  

\begin{theorem}\label{thm:Reduction_to_2PF}
    If the covariance $C$ is reflection positive, meaning that
    \begin{equation}
        \langle \Theta f,\, C\, f \rangle_{L^2} \geq 0\qquad \forall f\in \mathcal{S}(\reals^{d}_+),
    \end{equation}
    then the corresponding family of Schwinger functions is
    Osterwalder--Schrader positive in the sense of~\Cref{def:OS-positivity}.
\end{theorem}
\begin{proof}
    Let $F=(f_n)_{n\ge 0}\in \mathfrak B_+$. By charge conservation,
    \begin{equation}
        (F,F)_{\mathrm{OS}} = \sum_{n\ge 0} S^\pm_{2n}(\Theta \bar f_n\otimes f_n),
    \end{equation}
    so it suffices to show that each term is non-negative. The claim is trivial for $n=0$, so we fix $n\ge 1$. 

    First, suppose that $f_n \in \mathcal S(\reals^d_+)^{\otimes n}\subset \mathcal S(\reals^d_+)^{\hat\otimes n}$. Then $f_n$ can be written as a finite sum 
    \begin{equation}
        f_n = \sum_\alpha u_\alpha\, f^\alpha_1\otimes\dots\otimes f^\alpha_n, \qquad f^\alpha_i\in \mathcal S(\reals^d_+), \quad u_\alpha\in \complexes.
    \end{equation}
    By Wick's theorem,
    \begin{equation}
        S^\pm_{2n}(\Theta \bar f_n\otimes f_n) = \sum_{\alpha ,\beta } \bar{u}_\alpha u_\beta\; M^{\pm}_{\alpha\beta},
    \end{equation}
    where $M^{\pm}_{\alpha\beta} = \mathrm{det}^\pm \big( \langle\Theta f^\alpha_i, C f^\beta_j\rangle_{1\leq i,j\leq n}\big)$. We claim that $M^\pm$ is positive semidefinite. 
    
    By assumption, $\langle\Theta f, C g\rangle$ is positive semidefinite on $\mathcal S(\reals^d_+)$. It induces a positive semidefinite form on the $n$-th exterior power $\mathcal{S}(\reals^d_+)^{\wedge n}$ in the fermionic case, and on the $n$-th symmetric power $\mathcal{S}(\reals^d_+)^{\odot n}$ in the bosonic case. On decomposable tensors, said form is given by
    \begin{equation}
        \langle f_1 \wedge\dots\wedge  f_n, g_1 \wedge\dots\wedge  g_n\rangle = \mathrm{det} \big( \langle\Theta f_i, C g_j\rangle_{1\leq i,j\leq n}\big)
    \end{equation} 
    and 
    \begin{equation}
        \langle f_1 \odot\dots\odot  f_n,  g_1 \odot\dots\odot  g_n\rangle = \mathrm{perm} \big( \langle\Theta f_i, C g_j\rangle_{1\leq i,j\leq n}\big)
    \end{equation} respectively. Hence, $M^{\pm}_{\alpha\beta}$ is precisely the Gram matrix of the family $({f}^\alpha_1 \wedge\dots\wedge {f}^\alpha_n)_\alpha$ in the fermionic case, and of $({f}^\alpha_1 \odot\dots\odot {f}^\alpha_n)_\alpha$ in the bosonic case, and therefore positive semidefinite.

    Finally, since the algebraic tensor product is dense in $\mathcal S(\reals^d_+)^{\hat\otimes n}$, the general case follows by continuity. Thus, $(F,F)_{\mathrm{OS}} \ge 0$ for every $F\in \mathfrak B_+$.
\end{proof}

\section{Conclusion and Outlook}

We have been concerned with reflection positivity of quantum field theories in the continuum. Reflection positivity is a fragile symmetry that usually does not go well with regularizations, as our results also reconfirm. The major exception is the lattice regularization, where reflection positivity and the lattice subgroup of the Euclidean symmetry are conserved. The lattice approximation also has the advantage that it regularizes gauge theories in a gauge-invariant way, and independently of a perturbative expansion. Our main theorem confirms that RP is not present for polynomial regularizations of the Dirac covariance. While this is essentially a negative result, a way to utilize some of our techniques is to do a regularization that keeps positivity with respect to reflections in one type of planes, i.e., using a regulator that is no longer Euclidean invariant. The onus is then to prove that Euclidean invariance is restored in the limits as the regulator is removed. This is the subject of ongoing work. 

\begin{appendices}
\section{Reduction to the Covariance in the Real Scalar Case}
For a real scalar field, the proof that reflection positivity of the covariance implies reflection positivity of the Osterwalder--Schrader form is slightly more involved than in the complex case, since the Schwinger functions are not block diagonal in particle number. (The case of real fermions is almost identical, with pfaffians replacing hafnians and determinants replacing permanents.)

We write

\begin{equation}
  C(f,g)=\langle\bar f,Cg\rangle_{L^2}.
\end{equation}
In this case, Wick's theorem gives $S_{2n+1}=0$ and
\begin{equation}
\begin{aligned}
  S_{2n}(f_1\otimes\cdots\otimes f_{2n})
  &=
  \sum_{p\in\operatorname{Pair}(2n)}
  \prod_{\{i,j\}\in p}C(f_i,f_j)\\
  &=
  \operatorname{Haf}
  \bigl[C(f_i,f_j)\bigr]_{i,j=1}^{2n}.
\end{aligned}
\end{equation}
We collect the family $(S_n)_{n\geq0}$ into the Schwinger functional
\begin{equation}
  \omega_C:\mathfrak B\longrightarrow\mathbb C,
  \qquad
  \omega_C(F)
  =
  \sum_{n\geq0}S_n(f_n),
  \qquad
  F=(f_n)_{n\geq0}.
\end{equation}
The Osterwalder--Schrader form can be written in terms of the Schwinger functional:
\begin{equation}
  (F,G)_{\mathrm{OS}}
  =
  \omega_C\bigl((\Theta\bar F)G\bigr),
  \qquad
  F,G\in\mathfrak B_+.
\end{equation}
We call $C$ reflection invariant if 
\begin{equation}
  C(\Theta f, g) = \overline{C(\Theta g, f)}
\end{equation}
or equivalently $\Theta C \Theta = C$. A reflection invariant covariance $C$ is called reflection positive if
\begin{equation}
  C(\Theta f, f) \geq 0
  \qquad
  \forall f\in\mathcal S(\mathbb R_+^d).
\end{equation}
\begin{theorem}\label{thm:reduction-real-case}
If $C$ is a reflection positive covariance, then
\begin{equation}
  (F,F)_{\mathrm{OS}}\geq0
  \qquad
  \forall F\in\mathfrak B_+.
\end{equation}
\end{theorem}

Unlike in the complex case, the Osterwalder--Schrader form is not block
diagonal in particle number, since Wick contractions may occur within
either tensor factor. We will use Wick ordering to remove internal contractions.

Let
\begin{equation}
  \varepsilon_0:\mathfrak B\longrightarrow\mathbb C,
  \qquad
  F\mapsto \varepsilon_0(F)=f_0,
  \qquad
  F=(f_n)_{n\geq0},
\end{equation}
denote projection onto the degree-zero component. Define the contraction
operator on decomposable tensors by
\begin{equation}
  \Delta_C(f_1\otimes\cdots\otimes f_n)
  =
  2\sum_{1\leq i<j\leq n}
  C(f_i,f_j)\,
  f_1\otimes\cdots
  \widehat{f_i}\cdots
  \widehat{f_j}\cdots
  \otimes f_n,
\end{equation}
and extend it by continuity to the completed tensor powers.

Since $\Delta_C$ lowers degree by two, it is locally nilpotent on
$\mathfrak B$. 
(Recall that by definition of the algebraic direct sum, $\mathfrak B$ contains only finite sequences of test functions.) 
Consequently, both $e^{\frac12\Delta_C}$ and
$e^{-\frac12\Delta_C}$ are well-defined because their series terminate on every
element of $\mathfrak B$. In terms of $\Delta_C$, Wick's theorem is
equivalently expressed as
\begin{equation}
  \omega_C
  =
  \varepsilon_0\circ e^{\frac12\Delta_C}.
\end{equation}

\begin{lemma}[Properties of Wick ordering]
\label{lem:wick-ordering}
The Wick-ordering map
\begin{equation}
  W_C\colon\mathfrak B\longrightarrow\mathfrak B,
  \qquad W_C=e^{-\frac12\Delta_C}
\end{equation}
is a linear automorphism of $\mathfrak B$ with inverse
\begin{equation}
  W_C^{-1}=e^{\frac12\Delta_C}.
\end{equation}
It restricts to an automorphism of $\mathfrak B_+$ and, if $C$ is
reflection invariant, i.e. $\Theta C \Theta =C$,     satisfies
\begin{equation}
  W_C(\Theta\bar F)
  =
  \Theta\overline{W_CF}.
\end{equation}
\end{lemma}

\begin{proof}
Since $\Delta_C$ is locally nilpotent, the exponential series terminate
on every element of $\mathfrak B$. Hence
\begin{equation}
  e^{-\frac12\Delta_C}e^{\frac12\Delta_C}
  =
  e^{\frac12\Delta_C}e^{-\frac12\Delta_C}
  =
  \operatorname{id}_{\mathfrak B},
\end{equation}
which proves the first claim.

Since $\Delta_C$ preserves $\mathfrak B_+$, so do
$e^{-\frac12\Delta_C}$ and $e^{\frac12\Delta_C}$. Thus $W_C$ restricts
to an automorphism of $\mathfrak B_+$.

Finally, reflection invariance of $C$ gives
\begin{equation}
  \Delta_C(\Theta\bar F)
  =
  \Theta\overline{\Delta_CF}.
\end{equation}
Applying this identity term by term to the finite exponential series
yields
\begin{equation}
  W_C(\Theta\bar F)
  =
  \Theta\overline{W_CF}.
\end{equation}
\end{proof}

\begin{lemma}[Wick-ordered Gaussian moments]\label{lem:Permanent_after_Wick}
    Let $m,n \in \mathbb{N}_0$, and let $F_n = f_1\otimes \cdots \otimes f_n, G_m = g_1\otimes \cdots g_m$ with $f_i, g_j \in \mathcal{S}(\reals^d_+)$. Then
    \begin{equation}
        \omega_C\left(W_C(F_n) W_C(G_m)\right) =\begin{cases}
        \mathrm{perm}\left(C(f_i,g_j)\right)_{i,j=1}^n, \quad &m=n\\
        0, &m\neq n.
        \end{cases}
    \end{equation}
\end{lemma}
\begin{proof}
Let $\mu\colon\mathfrak B\otimes\mathfrak B\longrightarrow\mathfrak B$
denote multiplication in the Borchers algebra. On
$\mathfrak B\otimes\mathfrak B$, define the cross-contraction by
\begin{align}
  \Delta_C^{(12)}(F_n\otimes G_m)
  &=
  \sum_{i=1}^n\sum_{j=1}^m C(f_i,g_j)\\
  &\quad\times
  \bigl(
    f_1\otimes\cdots\widehat{f_i}\cdots\otimes f_n
  \bigr)
  \otimes
  \bigl(
    g_1\otimes\cdots\widehat{g_j}\cdots\otimes g_m
  \bigr), 
\end{align}

where the hats once again indicate omission. We will make use of the identity
\begin{equation}
    \Delta_C\mu
    =\mu\bigl(
    \Delta_C\otimes 1
    +1\otimes\Delta_C
    +2\Delta_C^{(12)}
  \bigr). 
\end{equation}

The three operators on the right commute. Recalling that
\begin{equation}
  \omega_C=\varepsilon_0e^{\frac12\Delta_C},
  \qquad
  W_C=e^{-\frac12\Delta_C},  
\end{equation}
we obtain
\begin{align}
  \omega_C\bigl(W_C(F_n)W_C(G_m)\bigr)
  &=
  \varepsilon_0e^{\frac12\Delta_C}
  \mu\bigl(W_C(F_n)\otimes W_C(G_m)\bigr)\\
  &=
  \varepsilon_0\mu\,
  e^{\Delta_C^{(12)}}(F_n\otimes G_m)\\
  &=
  (\varepsilon_0\otimes\varepsilon_0)
  e^{\Delta_C^{(12)}}(F_n\otimes G_m).
\end{align}
Here we used the identity
\begin{equation}
  \varepsilon_0\mu
  =
  \varepsilon_0\otimes\varepsilon_0.
\end{equation}

Since $\Delta_C^{(12)}$ lowers the degree in both tensor factors by
one, the last expression vanishes unless $m=n$. Suppose that $m=n$.
Then only the $n$-th term of the exponential contributes, so that
\begin{align}
  \omega_C\bigl(W_C(F_n)W_C(G_m)\bigr)
  &=
  \frac{1}{n!}
  (\varepsilon_0\otimes\varepsilon_0)
  \bigl(\Delta_C^{(12)}\bigr)^n(F_n\otimes G_n)\\
  &=  \frac{1}{n!}\sum_{\pi, \sigma\in\mathfrak S_n}
  \prod_{i=1}^n C(f_{\sigma(i)},g_{\pi(i)})\\
  &=  \sum_{\pi\in\mathfrak S_n}
  \prod_{i=1}^n C(f_i,g_{\pi(i)})\\
  &=\operatorname{perm}
  \bigl(C(f_i,g_j)\bigr)_{i,j=1}^n.
\end{align}
\end{proof}

We can now prove \Cref{thm:reduction-real-case}:
\begin{proof}
By~\Cref{lem:wick-ordering}, $W_C$ restricts to an automorphism of
$\mathfrak B_+$ and intertwines Wick ordering with reflection. Hence,
for every $G\in\mathfrak B_+$,
\begin{equation}
  (W_CG,W_CG)_{\mathrm{OS}}
  =
  \omega_C\bigl(W_C(\Theta\bar G)\,W_CG\bigr)
  \geq0,
\end{equation}
since by \Cref{lem:Permanent_after_Wick} the pullback of the Osterwalder--Schrader form under $W_C$ has the same Gram-matrix representation as in the complex scalar case and is therefore positive semidefinite.

Now let $F\in\mathfrak B_+$ and set $G=W_C^{-1}F\in\mathfrak B_+$. Then
\begin{equation}
  (F,F)_{\mathrm{OS}}
  =
  (W_CG,W_CG)_{\mathrm{OS}}
  \geq0.
\end{equation}
\end{proof}
\section{Exponential Regulators}
In this section, we examine reflection positivity for covariance operators that
contain analytic functions of the Laplacian rather than only rational functions.
For simplicity, we restrict attention to the scalar case. A representative
family of nonlocal covariances respecting $O(d)$-invariance is
\begin{equation}
  C(-\Delta)=(-\Delta+m^2)^{-1}\rho(\Delta),
\end{equation}
where \(\rho\) is analytic and satisfies suitable ultraviolet decay conditions.

We consider the momentum-space propagator
\begin{equation}
  \hat{C}_{\Lambda,m}(p^2)
  =
  \frac{e^{-p^2/\Lambda^2}}{p^2+m^2},
  \qquad \Lambda,m\in\mathbb R_+.
\end{equation}
Propagators of this form are frequently used as ultraviolet regulators. The
Gaussian factor obstructs the usual contour argument for reflection
positivity: it grows as \(p_0\to\pm i\infty\). Equivalently, absorbing the
factor into the test function produces convolution with a Gaussian in position
space, which does not preserve positive-time support. These observations
suggest that the regulated covariance is not reflection positive. We now
construct a counterexample directly.

Let \(C_{\Lambda,m}\) be the covariance operator with integral kernel
\begin{equation}
  C_{\Lambda,m}(x-y)
  =
  \int_{\mathbb R^d}\mathrm d^d p\,
  e^{ip\cdot(x-y)}\hat C_{\Lambda,m}(p^2).
\end{equation}

\begin{theorem}
  Let \(\Lambda>m\). Then \(C_{\Lambda,m}\) is not reflection positive: there
  exists \(f\in\mathcal S(\mathbb R_+^d)\) such that
  \begin{equation}
    I_{\Lambda,m}[f]
    =
    \left\langle\Theta f,C_{\Lambda,m}f\right\rangle_{L^2}
    <0.
  \end{equation}
\end{theorem}

\begin{proof}
  In Fourier space,
  \begin{equation}
    I_{\Lambda,m}[f]
    =
    \int_{\mathbb R^d}\mathrm d^d p\,
    \frac{e^{-p^2/\Lambda^2}}{p^2+m^2}
    \overline{\hat f(-p_0,\vec p)}
    \hat f(p_0,\vec p).
  \end{equation}
  This quadratic form extends continuously to \(L^2(\mathbb R_+^d)\). Since
  \(\mathcal S(\mathbb R_+^d)\) is dense in that space, it is enough to find an
  \(L^2\)-function of positive-time support for which the form is negative.

  Set
  \begin{equation}
    \omega_{\vec p}=\sqrt{|\vec p|^2+m^2}
  \end{equation}
  and make the ansatz that $f$ factors as
  \begin{equation}
    \hat f(p_0,\vec p)
    =
    \hat h(|\vec p|)\,
    \hat u\!\big(\tfrac{p_0}{\omega_{\vec p}}\big),
  \end{equation}
  where we want \(u\in L^2(\mathbb R_+)\). Substituting
  \(q=p_0/\omega_{\vec p}\), passing to spherical coordinates in
  \(\vec p\), and writing
  \begin{equation}
    \lambda(R)=\frac{\Lambda}{\sqrt{R^2+m^2}},
  \end{equation}
  gives
  \begin{equation}
    I_{\Lambda,m}[f]
    =
    S_{d-2}\int_0^\infty
    \frac{R^{d-2}\,\mathrm dR}{\sqrt{R^2+m^2}}
    e^{-R^2/\Lambda^2}|\hat h(R)|^2
    J_{\lambda(R)}[u],
  \end{equation}
  where
  \begin{equation}
    S_{d-2}
    =
    \frac{2\pi^{(d-1)/2}}{\Gamma((d-1)/2)}
  \end{equation}
  and
  \begin{equation}
    J_\lambda[u]
    =
    \int_{\mathbb R}\mathrm dq\,
    \frac{e^{-q^2/\lambda^2}}{q^2+1}
    \overline{\hat u(-q)}\hat u(q).
  \end{equation}

Choose $u(x)=(2-3x^2)e^{-x}\mathbf{1}_{(0,\infty)}(x)$. A direct calculation, using the recurrence
\begin{equation}
    2nI_{n+1}=2I_{n-1}+(2n-3)I_n,
    \qquad
    I_n=\int_{\mathbb{R}}\frac{e^{-q^2}}{(1+q^2)^n}\,dq,
\end{equation}
gives
\begin{equation} 
    J_1[u]
    =
    \frac{-50\sqrt{\pi}+21\pi e\,\operatorname{erfc}(1)}{80}
    <0,
\end{equation}
where the last inequality follows from
\begin{equation}
e\sqrt{\pi}\,\operatorname{erfc}(1)<1.
\end{equation}
Since $\Lambda>m$, define
\begin{equation}
R_*=\sqrt{\Lambda^2-m^2}>0,
\end{equation}
so that $\lambda(R_*)=1$. The map
\begin{equation}
R\longmapsto J_{\lambda(R)}[u]
\end{equation}
is continuous and negative at $R=R_*$. Hence, there exists an open interval $V\Subset(0,\infty)$ containing $R_*$ such that
\begin{equation}
J_{\lambda(R)}[u]<0
\qquad\text{for all }R\in V.
\end{equation}
Choose $0\neq\widehat{h}\in C_c^\infty(V)$. Since all other factors in the radial integral are positive, it follows that
\begin{equation}
I_{\Lambda,m}[f]<0.
\end{equation}
\end{proof}

\end{appendices}

\noindent
{\bf Acknowledgement. } This work is based in part on the first author's 2023 bachelor's thesis. This work is supported by Deutsche Forschungsgemeinschaft (DFG, German Research Foundation) under Germany's Excellence Strategy  EXC-2181/1 - 390900948 (the Heidelberg STRUCTURES Cluster of Excellence).

\printbibliography

@article{FSS1976,
  author  = {Fr{\"o}hlich, J. and Simon, B. and Spencer, T.},
  title   = {Infrared bounds, phase transitions and continuous symmetry breaking},
  journal = {Communications in Mathematical Physics},
  year    = {1976},
  volume  = {50},
  number  = {1},
  pages   = {79--95},
  doi     = {10.1007/BF01608557}
}

@article{FILS1,
  author  = {Fr{\"o}hlich, J. and Israel, R. and Lieb, E. H. and Simon, B.},
  title   = {Phase transitions and reflection positivity. {I}. General theory and long range lattice models},
  journal = {Communications in Mathematical Physics},
  year    = {1978},
  volume  = {62},
  number  = {1},
  pages   = {1--34},
  doi     = {10.1007/BF01940327}
}

@article{FILS2,
  author  = {Fr{\"o}hlich, J. and Israel, R. B. and Lieb, E. H. and Simon, B.},
  title   = {Phase transitions and reflection positivity. {II}. Lattice systems with short-range and {Coulomb} interactions},
  journal = {Journal of Statistical Physics},
  year    = {1980},
  volume  = {22},
  number  = {3},
  pages   = {297--347},
  doi     = {10.1007/BF01014646}
}

@article{JaffeMajorana1,
  author  = {Jaffe, A. and Pedrocchi, F. L.},
  title   = {Reflection positivity for {Majoranas}},
  journal = {Annales Henri Poincar{\'e}},
  year    = {2015},
  volume  = {16},
  number  = {1},
  pages   = {189--203},
  doi     = {10.1007/s00023-014-0311-y}
}

@article{JaffeParafermions,
  author  = {Jaffe, A. and Pedrocchi, F. L.},
  title   = {Reflection positivity for {Parafermions}},
  journal = {Communications in Mathematical Physics},
  year    = {2015},
  volume  = {337},
  number  = {1},
  pages   = {455--472},
  doi     = {10.1007/s00220-015-2340-x}
}

@article{OS1973,
  author  = {Osterwalder, K. and Schrader, R.},
  title   = {Euclidean {Fermi} fields and a {Feynman--Kac} formula for boson--fermion models},
  journal = {Helvetica Physica Acta},
  year    = {1973},
  volume  = {46},
  number  = {3},
  pages   = {277--302},
  doi     = {10.5169/seals-114484}
}

@article{OS1975,
  author  = {Osterwalder, Konrad and Schrader, Robert},
  title   = {{Axioms for Euclidean Green's Functions II}},
  year    = {1975},
  journal = {Commun. math. Phys},
  pages   = {281--305},
  volume  = {42},
  doi     = {10.1007/BF01608978}
}

@article{JaffeMajorana2,
  author  = {Jaffe, A. and Janssens, B.},
  title   = {Characterization of reflection positivity: {Majoranas} and spins},
  journal = {Communications in Mathematical Physics},
  year    = {2016},
  volume  = {346},
  number  = {3},
  pages   = {1021--1050},
  doi     = {10.1007/s00220-015-2545-z}
}

@book{Strichartz1994ATransforms,
  title     = {{A guide to distribution theory and Fourier transforms}},
  year      = {2003},
  author    = {Strichartz, Robert S.},
  pages     = {115--116},
  publisher = {World Scientific},
  doi       = {10.1142/5314}
}

@book{Bogolubov1990GeneralTheory,
  author    = {Bogolubov, N. N. and Logunov, A. A. and Oksak, A. I. and Todorov, I. T.},
  title     = {General Principles of Quantum Field Theory},
  series    = {Mathematical Physics and Applied Mathematics},
  volume    = {10},
  edition   = {1},
  year      = {1990},
  publisher = {Kluwer Academic Publishers},
  address   = {Dordrecht},
  doi       = {10.1007/978-94-009-0491-0},
  note      = {Translated from the Russian by G. G. Gould}
}

@book{Hormander2007AnalysisIII,
  author    = {H{\"o}rmander, Lars},
  title     = {The Analysis of Linear Partial Differential Operators {III}: Pseudo-Differential Operators},
  series    = {Classics in Mathematics},
  year      = {2007},
  publisher = {Springer},
  address   = {Berlin, Heidelberg},
  isbn      = {978-3-540-49937-4},
  doi       = {10.1007/978-3-540-49938-1}
}

@article{Glimm1987QuantumPhysics,
  title     = {{Quantum Physics}},
  year      = {1987},
  journal   = {Springer New York, NY},
  author    = {Glimm, James and Jaffe, Arthur},
  publisher = {Springer New York},
  doi       = {10.1007/978-1-4612-4728-9}
}

@book{Rudin1987RealAnalysis.,
  title     = {{Real and Complex Analysis}},
  year      = {1987},
  author    = {Rudin, Walter},
  edition   = {3rd},
  publisher = {McGraw-Hill},
  isbn      = {0-07-054234-1}
}

@article{Jaffe2008ReflectionMonotonicity,
  title     = {{Reflection positivity and monotonicity}},
  year      = {2008},
  journal   = {Journal of Mathematical Physics},
  author    = {Jaffe, Arthur and Ritter, Gordon},
  number    = {5},
  month     = {5},
  pages     = {052301},
  volume    = {49},
  url       = {https://aip.scitation.org/doi/abs/10.1063/1.2907660},
  doi       = {10.1063/1.2907660},
  issn      = {0022-2488}
}

@article{Saueressig2018_RP,
  title    = {Reflection positivity in higher derivative scalar theories},
  author   = {Arici, Francesca and Becker, Daniel and Ripken, Chris and Saueressig, Frank and van Suijlekom, Walter D.},
  journal  = {Journal of Mathematical Physics},
  volume   = {59},
  number   = {8},
  eid      = {082302},
  year     = {2018},
  month    = {08},
  issn     = {0022-2488},
  doi      = {10.1063/1.5027231},
  url      = {10.1063/1.5027231}
}

@article{Borchers1962,
  title    = {On structure of the algebra of field operators},
  author   = {Borchers, H.-J.},
  journal  = {Il Nuovo Cimento (1955--1965)},
  volume   = {24},
  number   = {2},
  pages    = {214--236},
  year     = {1962},
  doi      = {10.1007/BF02745645}
}

\end{document}